\documentclass[letterpaper, 10 pt, conference]{ieeeconf}  

\IEEEoverridecommandlockouts                              

\usepackage{amsthm}
\usepackage{amsmath} 
\usepackage{amssymb}  
\usepackage{caption}
\usepackage{color}
\usepackage{empheq} 
\usepackage{epsfig} 
\usepackage{graphicx} 
\usepackage{subcaption}
\usepackage{times} 
\usepackage{wrapfig}
\usepackage{multicol}
\usepackage[caption=false]{subfig}
\usepackage{cite}
\usepackage{relsize}
\usepackage[capitalise]{cleveref}
\usepackage{array}
\usepackage[]{algorithm2e}

\renewcommand{\vec}[1]{\mathbf{#1}}
\newcolumntype{M}{>{\centering\arraybackslash}m{\dimexpr.25\linewidth-2\tabcolsep}}

\graphicspath{{Media/}}
\newtheorem{theorem}{Theorem}
\newtheorem{proposition}{Proposition}

\title{\LARGE \bf
Optimal Disease Outbreak Detection in a Community Using Network Observability}

\author{Atiye Alaeddini$^{1}$ and Kristi A. Morgansen$^{2}$
\thanks{$^{1}$Atiye Alaeddini is with the William E. Boeing Department of Aeronautics and Astronautics,
        University of Washington, Seattle, WA, 98195-2400
        {\tt\small atiye@u.washington.edu}}%
\thanks{$^{2}$Kristi A. Morgansen is with the William E. Boeing Department of Aeronautics and Astronautics, University of Washington,
        Seattle, WA, 98195-2400
        {\tt\small morgansen@aa.washington.edu}}%
\thanks{This work was supported in part by ONR MURI N000141010952.}%
}

\begin{document}

\maketitle
\thispagestyle{empty}
\pagestyle{empty}

\begin{abstract}

Given a network, we would like to determine which subset of nodes should be measured by limited sensing facilities to maximize information about the entire network. The optimal choice corresponds to the configuration that returns the highest value of a measure of observability of the system. Here, the determinant of the inverse of the observability Gramian is used to evaluate the degree of observability. Additionally, the effects of changes in the topology of the corresponding graph of a network on the observability of the network are investigated. The theory is illustrated on the problem of detection of an epidemic disease in a community. The purpose here is to find the smallest number of people who must be examined to predict the number of infected people in an arbitrary community. Results are demonstrated in simulation.

\end{abstract}

\section{INTRODUCTION}

Analysing complex interconnected systems, such as the Internet, social networks, and biological networks has been the subject of increasing research attention in recent years. A key characteristic of system analysis is being able to determine information about the system state using measurement data. This observability analysis, as applied to a network, aims to predict the behaviour of each individual in the network based on a set of measured properties and on the local rules governing the state of individual vertices. Observability-based design of a multi-agent network can be found in some recent works such as \cite{pequito2014optimal, pequito2014design}. In these works, the objective is to find the structure of the communication graph to guarantee observability (unobservability) in \cite{pequito2014optimal} (\cite{pequito2014design}). The authors of these works studied the effect of different topologies of the communication graph and of the communication weights between nodes in a network on the rank of the observability matrix.

Social networks have the longest history in the field of study of real-world networks where studying interactions in human behaviour has received significant attention. For instance, it is of interest to consider the process of decision-making in a group of people \cite{srivastava2014collective} or to study the process of creating different colonies in a human group based on the level of trust between the individuals in that group \cite{xia2015structural}. In these studies, data collection is usually carried out by querying a subset of participants in the network directly using questionnaires or interviews. Such methods are labour-intensive, and the size of the network that can be observed is limited. Therefore the problem of optimal choice of nodes to be observed is critical. Another application of high relevance is electric power grid management. An electric power grid is a network of interconnected high-voltage transmission lines spanning a country or a portion of a country. Identification of abnormal patterns of behaviour in power distribution networks is very important and has been the subject of many studies recently, e.g.\ \cite{manandhar2014detection} and \cite{giraldo2014delay}. Finally, some biological systems can be represented as networks. Genetic regulatory networks, blood vessel networks, and neural networks are classes of biological network systems have been receiving recent attention in the research literature. 

In recent works, observability criteria have been used to find the best sensor location \cite{Krener09, mathews2007reliability, Serpas13}. In many cases, tests of observability can reveal useful information about system structure that can be leveraged to design more effective or more efficient estimators. For example, some choices of sensor placement may lead to faster estimator convergence times in nonlinear systems \cite{Singh05}. However, nonlinear observability analysis is computationally expensive. The concept of using computational rather than analytical methods is first introduced in \cite{lall2002subspace} to evaluate the observability Gramian of a nonlinear system. The empirical observabiltiy Gramian has been used successfully in a number of contexts for design of improved sensing \cite{DeVries13, Singh05, Serpas13}. In order to perform an optimization based on the observability Gramian matrix, a scalar function of the matrix must be chosen. The smallest singular value \cite{Krener09, DeVries13}, the determinant \cite{Serpas13, Singh05, wouwer2000approach}, the trace \cite{Singh05}, and the spectral norm \cite{Singh05} of the observability Gramian are some of the criteria that have been used for observability-based optimization problems. The measure of observability used in this paper is the determinant of the inverse of the empirical observability Gramian.
%
 
The focus of this paper is studying a type of biological network called a virus spreading network. The purpose is to detect the spread of an epidemic disease in an arbitrary community by observing the infected/healthy status of a small number of people in that community. We show that if the graph of interaction between individuals in a network is connected, we do not need to examine every single individual to detect their states. Instead we can allocate sensing facilities on a few nodes, which are obtained from an optimization problem introduced here, and reconstruct the states of the entire network.

The rest of the paper is structured as follows. In \cref{sec:prelim}, we give notation and definitions. The model under study is introduced in \cref{sec:problem}. The observability-based optimal selection of observing nodes is presented in \cref{sec:SuperNode}. The numerical simulations for the problem of disease spreading in an arbitrary community are given in \cref{sec:goodwin}, and \cref{sec:conclusion} concludes the paper.

\section{BACKGROUND}
\label{sec:prelim}

\subsection{Graph Theory}

A network is a set of items, which we call nodes, with connections between them, called edges. The interaction between agents in a network is represented by a graph $\mathcal{G}=(V,E)$. Each agent in a network is denoted as a node, and the edges represent interaction links between agents. The number of nodes is assumed to be $N$, and the number of edges to be $M$. The node set, $V$, consists of all nodes in the network. The edge set, $E$, is comprised of pairs of nodes $\{i,j\}$, where nodes $i$ and $j$ are adjacent. The neighbourhood set, $\mathcal{N}_i$, of node $i$ is composed of all agents in $V$ adjacent to node $i$. The edges are encoded through the index mapping $\sigma$ such that $l=\sigma(i,j)$, if and only if edge $l$ connects nodes $i$ and $j$. The adjacency matrix is an $N \times N$ symmetric matrix with $\mathcal{A}_{ij}=1$ when $\{i,j\} \in E$, and $\mathcal{A}_{ij}=0$, otherwise. Here, we assume an undirected graph structure where $\{i,j\} \in E \Rightarrow \{j,i\} \in E$.

\subsection{System Dynamics and Observability}

In this work, we desire to maximize the amount of information we can gather about all the agents in the network by observing a small number of agents. Specifically, we would like to maximize the observability of the overall network by sharing limited information between the agents. To facilitate our study of network observability, specifically in the context of nonlinear process or measurement dynamics, we present a brief review of the relevant aspects of nonlinear observability that will be utilized here.

Consider a nonlinear system 
\begin{equation} \Sigma :\left\{ \begin{aligned}
 &    \dot{\vec{x}}=\vec{f}(\vec{x}), \ \  \vec{x} \in \mathbb{R}^n, \\
 &	 \vec{y}=\vec{h}(\vec{x}),  \ \ \vec{y} \in \mathbb{R}^m. \label{ODE} \end{aligned} \right. 
\end{equation}
Given the nonlinear system, $\Sigma$, one can linearize the nonlinear system about a given nominal trajectory, $\tilde{\vec{x}}(t)$. The linearization of this system is given by
\begin{equation}
    \delta \vec{\dot{x}}=F(t) \delta \vec{x}, \ \  \delta y=H(t) \delta \vec{x}, \label{linearized}
\end{equation}
where
\begin{equation}
    F(t) = \left. \frac{\partial \vec{f}}{\partial \vec{x}} \right |_{\vec{x}=\tilde{\vec{x}}}, \ \ H(t) = \left. \frac{\partial \vec{h}}{\partial \vec{x}} \right |_{\vec{x}=\tilde{\vec{x}}}. \label{Jacobi_mat}
\end{equation}
One approach for evaluating observability of nonlinear systems is to use the observability Gramian of a linearization of the system. 
If $\Phi(t)$ is the state transition matrix of the linearized system, then the local observability Gramian is
\begin{equation}
	W_O(\tilde{\vec{x}})=\int_{0}^{t_f} \Phi^T(t)H^T(t)H(t)\Phi(t) \mathrm{d}t \,.
\end{equation}
This formulation requires computation of Jacobian matrices. 
Furthermore, the Gramian of the linearized system will not always result in a good description of the behaviour of the original nonlinear system, specifically on wide operating ranges. To address these issues, we use the empirical local observability Gramian \cite{Krener09}. This tool provides an accurate description of a nonlinear system's observability, while it is much less computationally expensive than some other tools such as Lie algebra based approaches.

The concept of the empirical observability Gramian is related to the concept of output energy. Given a small perturbation $\epsilon > 0$ of the state, let $\vec{x}^{\pm i} (0)= \vec{x}(0) \pm  \epsilon \vec{e}_i$ and $\vec{h}^{\pm i}(t)$ be the corresponding output, where, $\vec{e}_i$ is the $i^{th}$ unit vector in $\mathbb{R}^n$. The empirical local observability Gramian at $\vec{x}(0)$ is the $n \times n$ matrix, $W_O$, whose $(i,j)$ component is 
\begin{equation}
	\frac{1}{4 \epsilon^2} \int_{0}^{t_f} \left( \vec{h}^{+i}(t)-\vec{h}^{-i}(t) \right)^T \left(\vec{h}^{+j}(t)-\vec{h}^{-j}(t) \right) \mathrm{d}t.\nonumber \label{EmpGram}
\end{equation}
It can be shown that if the system is smooth, then the empirical local observability Gramian converges to the local observability Gramian as $\epsilon \to 0$. Note that the perturbation, $\epsilon$, should always be chosen such that system stays in the region of attraction of the equilibrium point. As suggested by the authors of \cite{Krener09}, if the size of each state coordinate is of order one, then a reasonable choice of $\epsilon$ is order 0.01 or 0.001.

\section{PROBLEM DEFINITION}
\label{sec:problem}

\subsection{Problem Statement}

We assume a system of $N$ agents. Measurement of node $i$ is given by 
\begin{equation}
	y^i = x_i \,,\ \ i=1, \ldots, N. \label{imeasure}
\end{equation}
Now, let us assume there is a group of agents called observing nodes, $\bar{V}$, from which we want to maximize information measured from the network. Therefore, we have a subset of nodes $\bar{V} \subset V$ for which we want to maximize a measure of observability. Define a binary variable $\zeta \in \mathbb{R}^N$ such that $\zeta_i = 1$ if $V_i \in \bar{V}$ and $\zeta_i = 0$ otherwise. Then, the measurement is given by
\begin{equation}
	\vec{y} =\begin{bmatrix} \zeta_1 y^1 & \zeta_2 y^2 & \cdots &\zeta_N y^N \end{bmatrix}^T\,. \label{Measurement}
\end{equation}

Here, the objective is to determine a set of $r < N$ nodes such that if we observe those $r$ nodes, then we would be able to reconstruct the state of all nodes in the network. In order to find these nodes, we maximize the determinant of the observability Gramian by varying the location of the observing nodes in the network. Maximizing the determinant of the observability Gramian corresponds to a maximization of independence between outputs \cite{wouwer2000approach}. The problem can be formulated mathematically as:
\begin{equation}
\begin{aligned}
& \underset{\zeta_1, \zeta_2, \cdots, \zeta_N}{\text{maximize}}
& & \det(W_O)\,.
\end{aligned}
\end{equation}
Suppose the observation is partitioned into multiple sub-vectors from a number of individual sensors (similar to \eqref{Measurement}), then the total observability Gramian is the sum of the separate observability Gramian matrices, each obtained from having each individual sensor separately \cite{Krener09}. Therefore,
\begin{equation}
	W_O=\sum_{i=1}^{N} \zeta_i W_{O,i} \,,
\end{equation}
where $W_{O,i}$ is the observability Gramian assuming the observation to be $y= x_i$, i.e.\ we have a single sensor located on node $i$. Knowing this property of the observability Gramian matrix, the problem can be rewritten as a minimization problem, termed \emph{D-optimal design} in \cite{boyd2004convex}:
\begin{equation}
\begin{aligned}
& \underset{\zeta}{\text{minimize}}
& & \log \left[ \det\left( \sum_{i=1}^{N} \zeta_i W_{O,i} \right)^{-1} \right] \\
& \text{subject to}
& & \sum_{i=1}^{N} \zeta_i \leq r \, \\
&
& & \zeta_i \in \{0,1\},
\end{aligned} \label{MaxDet}
\end{equation}
where, $W_{O,i}$ is the observability Gramian obtained from measuring node $i$. 

\subsection{Virus Spreading Model in Networks}
\label{subsec:VirusModel}

In this paper, we formulate our problem using the setting proposed in \cite{van2013homogeneous} for virus spread in any network with $N$ nodes. This approach is one of the most popular epidemic models and is called the Susceptible-Infected-Susceptible (SIS) model \cite{anderson1992infectious}. In this model, each node in the network is either infected or healthy. An infected node, $i$, can infect its neighbours with an infection rate, $\beta_i$, and it is cured with curing rate, $\delta_i$. As a node is cured and healthy, it is again prone to the virus. The spread is modelled by an undirected network specified by a symmetric adjacency matrix, $\mathcal{A}$. The state of each node $i$ is described by the binary random variable $X_i(t) \in \{H,I\}$, i.e.\, the node at time $t$ in the network has two states: infected with probability $Pr[X_i(t) = I]$ and healthy with probability $Pr[X_i(t) = H]$. The evolution of the states is described by a Markov Process. The two possible state transitions of node $i$ are:
\begin{enumerate}
\item[(1)] Assume node $i$ is healthy at time $t$, i.e.\ $X_i(t) = H$. This node can switch to infected state over small time $\Delta t > 0$ with probability: $Pr[X_i(t+\Delta t)= I| X_i(t) = H]= \sum_{j\in \mathcal{N}_i} \mathcal{A}_{ij}\beta_i X_j(t) \Delta t + o(\Delta t)$.
\item[(2)] Assume node $i$ is infected at time $t$, i.e.\ $X_i(t) = I$. The probability of being recovered after small time $\Delta t > 0$ is: $Pr[X_i(t+\Delta t) = H| X_i(t) = I] = \delta_i \Delta t + o(\Delta t)$.
\end{enumerate}

Denoting $x_i(t) = Pr[X_i(t) = I]$ and considering that $Pr[X_i(t) = H] = 1-x_i(t)$, the Markov differential equation for node $i$ turns out to be nonlinear, $\dot{x}_i(t)=f_i(\vec{x})$, as follows:
\begin{equation}
	\dot{x}_i(t)=\beta_i \sum_{j=1}^N \mathcal{A}_{ij} x_j(t) - x_i(t) \left(\beta_i \sum_{j=1}^N \mathcal{A}_{ij} x_j(t) + \delta_i \right). \label{DiffNode_i}
\end{equation}
By defining $\vec{x} = \begin{bmatrix} x_1 & x_2 & \cdots & x_N \end{bmatrix}^T$, the differential equation \eqref{DiffNode_i} can be written in matrix form as
\begin{equation}
	\dot{\vec{x}}(t)= \left( B \mathcal{A} - D \right) \vec{x}(t) -\left( \sum_{j=1}^N e_j e_j^T \vec{x}(t) e_j^T \right) B \mathcal{A} \vec{x}(t) \,, \label{DiffNodes}
\end{equation}
where $B = \text{ diag}(\beta_j)$ and $D = \text{ diag}(\delta_j)$. In this paper, we assume that $\lambda_1 \left( B \mathcal{A} - D \right) < 0$, where $\lambda_1(\cdot)$ is the largest eigenvalue. It is proved in \cite{preciado2013optimal} that the linear dynamics system $\dot{\vec{x}}(t) = \left( B \mathcal{A} - D \right) \vec{x}(t)$ is an upper bound for the nonlinear dynamics system \eqref{DiffNodes}. Therefore, the disease-free equilibrium ($\vec{x} = 0$) is stable.

\subsection{Observability Gramian Computation}

As it was explained in \cref{subsec:VirusModel}, the model of virus spreading is a nonlinear dynamics system. 
The empirical observability Gramian, considering $y^k = x_k$, is an $N \times N$ matrix, $W_{O,k}$, whose $(i,j)$ component is 
\begin{equation}
	\frac{1}{4\epsilon^2} \int_{0}^{t_f} \left(x_k^{+i}(t)-x_k^{-i}(t)\right)^T \left(x_k^{+j}(t)-x_k^{-j}(t)\right) \mathrm{d}t. \label{EmpGram_i}
\end{equation}
In this case, the initial condition of the system should be perturbed in all directions of the states. In the case of a network with a large number of nodes, we need to solve the nonlinear differential equation \eqref{DiffNodes} $2N$ times. When the number of nodes is large, the computation of the empirical observability Gramian for all $N$ nodes is computationally expensive. We would like to be able to detect an epidemic by knowing only a subset of nodes in the network, termed the \emph{test set}. Here, we randomly select $s \leq N$ nodes, and perturb the initial condition in these $s$ directions. Then the observability Gramian, $W_{O}$, is an $s \times s$ matrix. Note that $r \ll s$, and the test set is chosen randomly.

The infection rate of node $i$, $\beta_i$, and the recovery rate of node $i$, $\delta_i$, can be changed by allocating preventative resources such as vaccinations and antidotes on this node. As the value of either infection rate or recovery rate changes, the observability Gramian matrices should be updated accordingly. 

The observability Gramian, $W$, of a linear system,
\begin{equation}
	\dot{\vec{x}} = A \vec{x}\, \ \ \vec{y} = C \vec{x}
\end{equation}
is obtained by solving for $W$ in
\begin{equation}
	A^T W + W A + C^T C =0\,.
\end{equation}
Therefore,
\begin{equation}
	A^T \left(\frac{\partial W}{\partial \chi}\right) + \left(\frac{\partial W}{\partial \chi}\right) A +\left(\frac{\partial A}{\partial \chi}\right)^T W + W \left(\frac{\partial A}{\partial \chi}\right)=0\,, \label{Lyap_beta_k}
\end{equation}
where $\chi$ is a variable of the dynamic system ($A$), and does not affect the observation ($C$). This equation can be solved to obtain $\displaystyle \left(\frac{\partial W}{\partial \chi}\right)$. 

\begin{proposition} \label{Linearizedxtilde}
Consider the adjacency matrix, $\mathcal{A}$, of a connected graph, and two sets of positive numbers $\{\beta_i\}_{i=1}^N$ and $\{\delta_i\}_{i=1}^N$. Assume $\tilde{\vec{x}}$ represents the infected/healthy status of a network at an arbitrary time $\tau$. Define $\tilde{\beta}_i = \left(1-\tilde{x}_i\right)\beta_i$ and $\tilde{B}=\text{diag}(\tilde{\beta}_i)$. Then the nonlinear system \eqref{DiffNodes} for any $\vec{x}$ close to $\tilde{\vec{x}}$ can be approximated as $\dot{\vec{x}}(t) = \left( \tilde{B} \mathcal{A} - D \right) \vec{x}(t)$.
\end{proposition}

\begin{proof}
Using the Taylor expansion of $f_i(\vec{x})$ in \eqref{DiffNode_i}, the linear approximation in the vicinity of $\tilde{\vec{x}}$ is given by
\begin{equation} \label{tylorApp_i}
	\begin{aligned}
	&f_i(\vec{x}) = f_i(\tilde{\vec{x}}) +\left. \sum_{j=1}^N \frac{\partial f_i}{\partial x_j}\right|_{\tilde{\vec{x}}} \left(x_j - \tilde{x}_j \right)+ \varsigma_i \\
	&= \tilde{\beta}_i \sum_{j=1}^N \mathcal{A}_{ij} x_j-\delta_i x_i + \varsigma_i \,,
	\end{aligned}
\end{equation}
where 
\begin{equation}
\begin{aligned}
	\varsigma_i &= \frac{-\beta_i}{2} \left(\vec{x} - \tilde{\vec{x}} \right)^T \vec{e}_i \vec{a}_i^T \left(\vec{x} - \tilde{\vec{x}} \right) \,,
\end{aligned}
\end{equation}
where $\vec{a}_i = \mathcal{A} \vec{e}_i$. If $\left(\vec{x} - \tilde{\vec{x}} \right) \to \vec{0}$, then $\varsigma_i \to 0$, and the dynamic system can be written in matrix form as
\begin{equation} \label{linApprox_f}
\dot{\vec{x}}(t) \approx \left( \tilde{B} \mathcal{A} - D \right) \vec{x}(t) \,.
\end{equation}
\end{proof}

\begin{theorem}
Consider the adjacency matrix, $\mathcal{A}$, of a connected graph, and two sets of positive numbers $\{\beta_i\}_{i=1}^N$ and $\{\delta_i\}_{i=1}^N$ such that $\lambda_1 \left( B \mathcal{A} - D \right) < 0$. Define the approximation error of \eqref{linApprox_f}, $\varsigma$, as
\begin{equation}
	\varsigma = | \vec{f}(\vec{x}) - \left( \tilde{B} \mathcal{A} - D \right) \vec{x}(t) |_1 \,.
\end{equation}
Then, $\varsigma < \frac{1}{2} \text{max} \{\delta_i\}_{i=1}^N \| \tilde{\Delta} \|^2$,
where $\tilde{\Delta}_i = \left|x_i-\tilde{x}_i \right|$.
\end{theorem}

\begin{proof}
Based on the definition of the approximation error, $\varsigma$, we have
\begin{equation}
\begin{aligned}
\varsigma =\sum_{i=1}^N \left| \varsigma_i \right| &\leq \frac{1}{2} \sum_{i=1}^N \beta_i \tilde{\Delta}_i \sum_{j=1}^N \mathcal{A}_{ij} \tilde{\Delta}_j = \frac{1}{2} \tilde{\Delta}^T B \mathcal{A} \tilde{\Delta} \\
& \leq \frac{1}{2} \lambda_1(B \mathcal{A}) \| \tilde{\Delta} \|^2 \,.
\end{aligned}
\end{equation}
Recall that we assumed $\lambda_1 \left( B \mathcal{A} - D \right) < 0$. Thus,
\begin{equation}
\lambda_1(B \mathcal{A}) < \text{max} \{\delta_i\}_{i=1}^N \,.
\end{equation}
Therefore $\varsigma < \frac{1}{2} \text{max} \{\delta_i\}_{i=1}^N \| \tilde{\Delta} \|^2$.

\end{proof}

\begin{theorem} \label{AtildeStable}
Consider the adjacency matrix $\mathcal{A}$ of a connected graph, and two sets of positive numbers $\{\beta_i\}_{i=1}^N$ and $\{\delta_i\}_{i=1}^N$ such that $\lambda_1 \left( B \mathcal{A} - D \right) < 0$. Then the linearized dynamic system $\dot{\vec{x}}(t) = \left( \tilde{B} \mathcal{A} - D \right) \vec{x}(t)$ is asymptotically stable.
\end{theorem}

\begin{proof}
Using the definition of $\tilde{B} = \text{diag} \left( \left( 1-\tilde{x}_i \right) \beta_i\right)$, the dynamics of node $i$ is given by
\begin{equation}
	\dot{x}_i(t) = \beta_i \sum_{j=1}^N \mathcal{A}_{ij} x_j(t)-\delta_i x_i(t) - \beta_i \tilde{x}_i \sum_{j=1}^N \mathcal{A}_{ij} x_j(t) \,.
\end{equation}
Since $\beta_i, \tilde{x}_i, x_j(t), \mathcal{A}_{ij} \geq 0$, then 
\begin{equation}
	\dot{x}_i(t) \leq \beta_i \sum_{j=1}^N \mathcal{A}_{ij} x_j(t)-\delta_i x_i(t) \,.
\end{equation}
Since $\lambda_1 \left( B \mathcal{A} - D \right) < 0$, then the system 
\begin{equation}
	\dot{\hat{x}}_i(t) = \beta_i \sum_{j=1}^N \mathcal{A}_{ij} \hat{x}_j(t) - \delta_i \hat{x}_i(t)
\end{equation}
is asymptotically stable. Given that $\dot{x}_i(t) \leq \dot{\hat{x}}_i(t)\,, \forall t$, we can conclude that the linearized dynamics system, $\dot{\vec{x}}(t) = \left( \tilde{B} \mathcal{A} - D \right) \vec{x}(t)$, should also be asymptotically stable.
\end{proof}

To study the effect of changes in infection and recovery rates of a node (i.e.\,, $\beta_k$ or $\delta_k$) on the observability Gramian, assume a change occurs at time $t = \tau$ where $\vec{x}(\tau)=\tilde{\vec{x}}$. Now, by using the result of \cref{Linearizedxtilde}, the dynamics of the system can be written as
\begin{equation} \label{Atilde}
	\dot{\vec{x}}(t) \approx \left( \tilde{B} \mathcal{A} - D \right) \vec{x}(t) = \tilde{A} \vec{x}(t)\,.
\end{equation}

\begin{proposition} \label{uniqueSol}
Given a positive semi-definite matrix, $W$, and the adjacency matrix, $\mathcal{A}$, of a connected graph, then the Lyapunov equation
\begin{equation} \label{LyapAtilde}
	\tilde{A}^T \left(\frac{\partial W}{\partial \chi}\right) + \left(\frac{\partial W}{\partial \chi}\right) \tilde{A} + \left(\frac{\partial \tilde{A}}{\partial \chi}\right)^T W + W \left(\frac{\partial \tilde{A}}{\partial \chi}\right) =0\,,
\end{equation}
where $\tilde{A}$ was defined in \eqref{Atilde}, has exactly one solution for $\displaystyle \left(\frac{\partial W}{\partial \chi}\right)$.
\end{proposition}

\begin{proof}
The proof is directly derived from the result of \cref{AtildeStable}, which demonstrates that $\tilde{A}$ is stable. Since $\tilde{A}$ is stable, we know that the Lyapunov operator defined by $\mathcal{L}(P) = \tilde{A}^T P + P \tilde{A}$ is non-singular and for a symmetric matrix, $Q$, the equation $\displaystyle \mathcal{L}(P)+Q = 0$ has exactly one solution for $P$. Since $\displaystyle \left[ \left(\frac{\partial \tilde{A}}{\partial \chi}\right)^T W + W \left(\frac{\partial \tilde{A}}{\partial \chi}\right)\right]$ is a symmetric matrix, thus the Lyapunov equation \eqref{LyapAtilde} has a unique solution.
\end{proof} 

Now, assume $W_{O}(\beta^0,\delta^0)$ is the empirical observability Gramian obtained from \eqref{EmpGram_i} for the initial values of infection and recovery rates, $\displaystyle \{\beta_i^0\}_{i=1}^N$ and $\displaystyle \{\delta_i^0\}_{i=1}^N$ respectively. Considering $\tilde{A} = \left( \tilde{B} \mathcal{A} - D \right)$, then 
\begin{equation} \label{Atilde_beta_delta}
\frac{\partial \tilde{A}}{\partial \beta_k} = \left(1-\tilde{x}_k \right) \vec{e}_k \vec{e}_k^T \mathcal{A} \,, \ \ \frac{\partial \tilde{A}}{\partial \delta_k} = - \vec{e}_k \vec{e}_k^T \,.
\end{equation}
%

In the case of a change in the infection rate of node $k$ (e.g.\,, because of vaccination of node $k$), or a change in the recovery rate of node $k$ (e.g.\,, because of allocating antidotes on node $k$), we do not need to re-calculate the empirical observability Gramian, $W_{O}$. Instead, we can substitute $\chi = \beta_k$ or $\chi = \delta_k$ in \eqref{LyapAtilde} and use \eqref{Atilde_beta_delta} to solve the corresponding Lyapunov equation. 
As given by \cref{uniqueSol}, there is exactly one solution of this Lyapunov equation which is used to obtain the updated observability Gramian matrix in the case of a change in the infection and recovery rates.

\section{OPTIMAL OBSERVING NODES ALLOCATION}
\label{sec:SuperNode}

The optimization problem \eqref{MaxDet} is a boolean nonlinear programming problem. Because of the binary constraint, this optimization problem is a non-convex problem. A common method for solving these types of optimization problems is computing a lower-bound on the optimal value of the non-convex problem. In general, there are two standard methods to solve these types of non-convex problems. 

The first method is relaxation, in which the non-convex constraint is replaced with a looser, but convex constraint. For example, in the case when $r$ is large, a good approximate solution of \eqref{MaxDet} can be found by ignoring, or relaxing, the constraint that the values of $\zeta_i$ are integers \cite{boyd2004convex}. 

The next method is Lagrangian relaxation, where we need to solve the convex dual problem. For example, for solving \eqref{MaxDet}, the boolean constraint can also be reformulated as $\zeta_i (\zeta_i-1)=0$, which is a quadratic equality constraint. Then we can solve the Lagrange dual of this problem. The optimal value of the relaxed problem provides a lower bound on the optimal value of the original optimization problem. 

Here, the \emph{Outer Approximation} technique, given in \cite{bonami2008algorithmic}, is used to solve our optimization problem. The advantage of this method compared to the two conventional relaxation methods discussed is that the algorithm converges to an optimal solution of \eqref{MaxDet} in a finite number of iterations. The convergence proof can be found in \cite{bonami2008algorithmic}. This algorithm uses linearization of the objective function and the constraints at different points to build a mixed integer linear programming relaxation of the problem. If $F(\zeta)$ refers to the cost function given in \eqref{MaxDet}, then we have
\begin{equation}
	 \nabla_j  F(\zeta) = -\text{trace} \left[ \left( \sum_{i=1}^{N} \zeta_i W_{O,i} \right)^{-1} W_{O,j} \right]\,.
\end{equation}
A lower bound solution of problem \eqref{MaxDet} is obtained by solving a Mixed Integer Linear Programming as follows. For any given set of points $T$, we can build a relaxation of \eqref{MaxDet} as:
\begin{equation}
\begin{aligned}
& \underset{\zeta}{\text{minimize}}
& & \alpha \\
& \text{subject to}
& & \nabla F(\zeta-\bar{\zeta})+F(\bar{\zeta}) \leq \alpha, \ \ \forall \bar{\zeta} \in T \\
&
& & \sum_{i=1}^{N} \zeta_i \leq r\,, \ \ \zeta \in \{0,1\}\,.
\end{aligned} \label{RelaxedMaxDet}
\end{equation}
The relaxation \eqref{RelaxedMaxDet} results in an iterative algorithm for solving the original problem \eqref{MaxDet}. This iterative algorithm basically relies on updating the set of linearization points, $T$. The algorithm starts with $T = \{\zeta^0\}$, where $\zeta^0$ is a feasible solution of the original problem \eqref{MaxDet} or of its continuous relaxed problem. Each iteration starts by solving \eqref{RelaxedMaxDet} to find a point $(\alpha^k,\zeta^k)$ and a lower bound $\alpha^k$ on the optimal value of \eqref{MaxDet}. Now, $\zeta^k$ is added to $T$. The algorithm stops when the difference between the lower bound linear approximation and the actual value of the cost function at that point becomes negligible. The algorithm is described in Algorithm \ref{OuterApprox}. 
\begin{algorithm}[!h]
 $Z^U:=+\infty$\;
 $Z^L:=-\infty$\;
 $\zeta^0:=$ optimal solution of \eqref{MaxDet}, replacing the integer constraint with continuous constraint $0 \leq \zeta^i \leq 1$\;
 $k:=1$\; 
 Choose a convergence tolerance $\epsilon_t$\;
 \While{$Z^U-Z^L > \epsilon_t$ and \eqref{RelaxedMaxDet} is feasible}{
  Let $(\hat{\alpha},\hat{\zeta})$ be the optimal solution of \eqref{RelaxedMaxDet}\;
  $\zeta^k:=\hat{\zeta}$\;
  $Z^L:=\hat{\alpha}$\;
  $Z^U:=\min\{Z^U,F(\zeta^k)\}$\;
  $T:=T \cup \{\zeta^k\}$\;
  $k:=k+1$\;
 }
 \caption{Outer approximation algorithm} \label{OuterApprox}
\end{algorithm}


\section{APPLICATION: DETECTING EPIDEMIC DISEASE TRANSMISSION}
\label{sec:goodwin}

The problem of containing virus spreading processes in a network has been studied in the literature, e.g. \cite{preciado2014optimal}. This modelling is an important step toward understanding the behaviour of spreading a disease in a community. However, given an arbitrary network, it remains a challenge to determine the requirements that make the entire network observable. Here, we consider a network of connected agents, and find the best selection of nodes which yield the state information of the entire network. In this section, we present simulation results that demonstrate the effect of graph configurations and demonstrate the corresponding theoretical analysis.

It is worthy to note here that the optimal value function in \eqref{MaxDet} depends on the value of $r$. If $r$ is small compared to $N$, the optimal solution of \eqref{MaxDet} is too large. If the optimal solution is large, we conclude that the determinant of the observability Gramian of the network with observing nodes obtained from the optimization problem is very small, which corresponds to a practically unobservable condition. Therefore, first we need to find the minimum number of observing nodes required for observability. Here, to find a minimum number of observing nodes, we set an upper-bound accepted value for the cost function, such that the solution of \eqref{MaxDet} is acceptable only if the solution is less than the upper-bound accepted value. In this paper, the upper-bound accepted value is set to zero, which is equivalent to $\det(W_O) \geq 1$.

Consider a fairly small group of people interacting with each other. This group is modelled by an undirected graph of 15 nodes (\cref{10Graph}). The nodes in the graph could be students of a class in an elementary school, and the edges denote the friendship between the children in this class. This graph could also be cities in a specific region which are connected through air transportation.

The initial healthy/infected status of each node and the infection and recovery rates are chosen randomly. Here, the purpose is to find a small number of observing nodes, if possible, such that we obtain the health data of the entire group. We first try to find a single observing node to obtain the status of the entire network. In the first example (\cref{Random10Graph}), we have a dense structure. This structure is an example of a network with many interactions between them. By solving the optimization problem, one can find that by observing node 6, the network becomes observable. Now, consider a network with much less interaction between the nodes (\cref{Sparse10Graph}). The calculation shows that in this case, we are not able to attain the state information of the entire network by observing only a single node. In the case of sparse structure with one observing node, the observability Gramian is close to be singular, $\det(W_O) \approx 0$. Therefore, the network is not observable. By increasing the number of observing nodes to $r=2$, the solution of the optimization problem becomes feasible, and hence, the network becomes observable.
\begin{figure}[!h]
        \centering
        \begin{subfigure}[b]{0.23\textwidth}
                \centering
                \includegraphics[width=\textwidth]{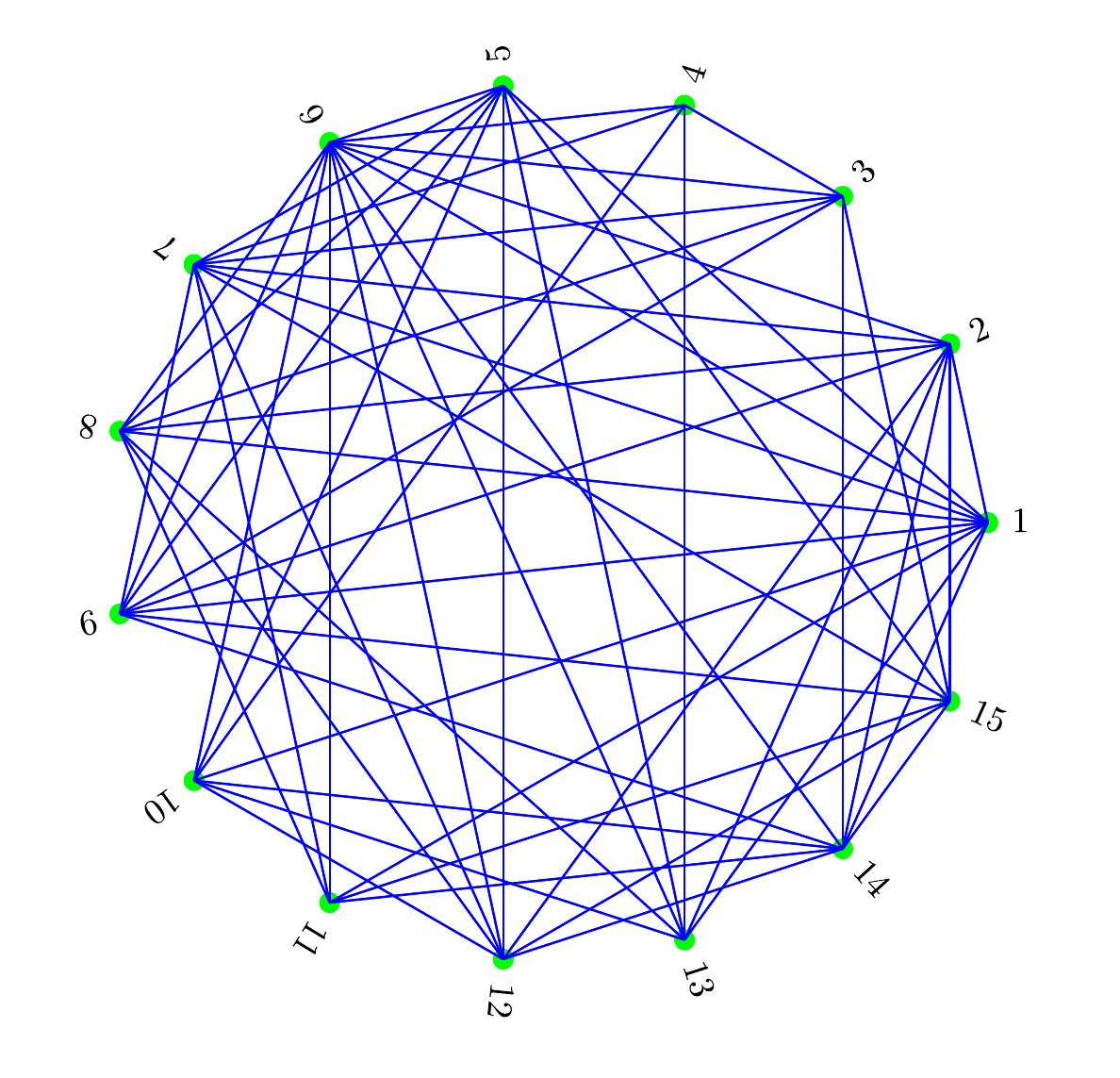}
                \caption{dense graph}
                \label{Random10Graph}
        \end{subfigure}%
        \quad 
        \begin{subfigure}[b]{0.23\textwidth}
                \centering
                \includegraphics[width=\textwidth]{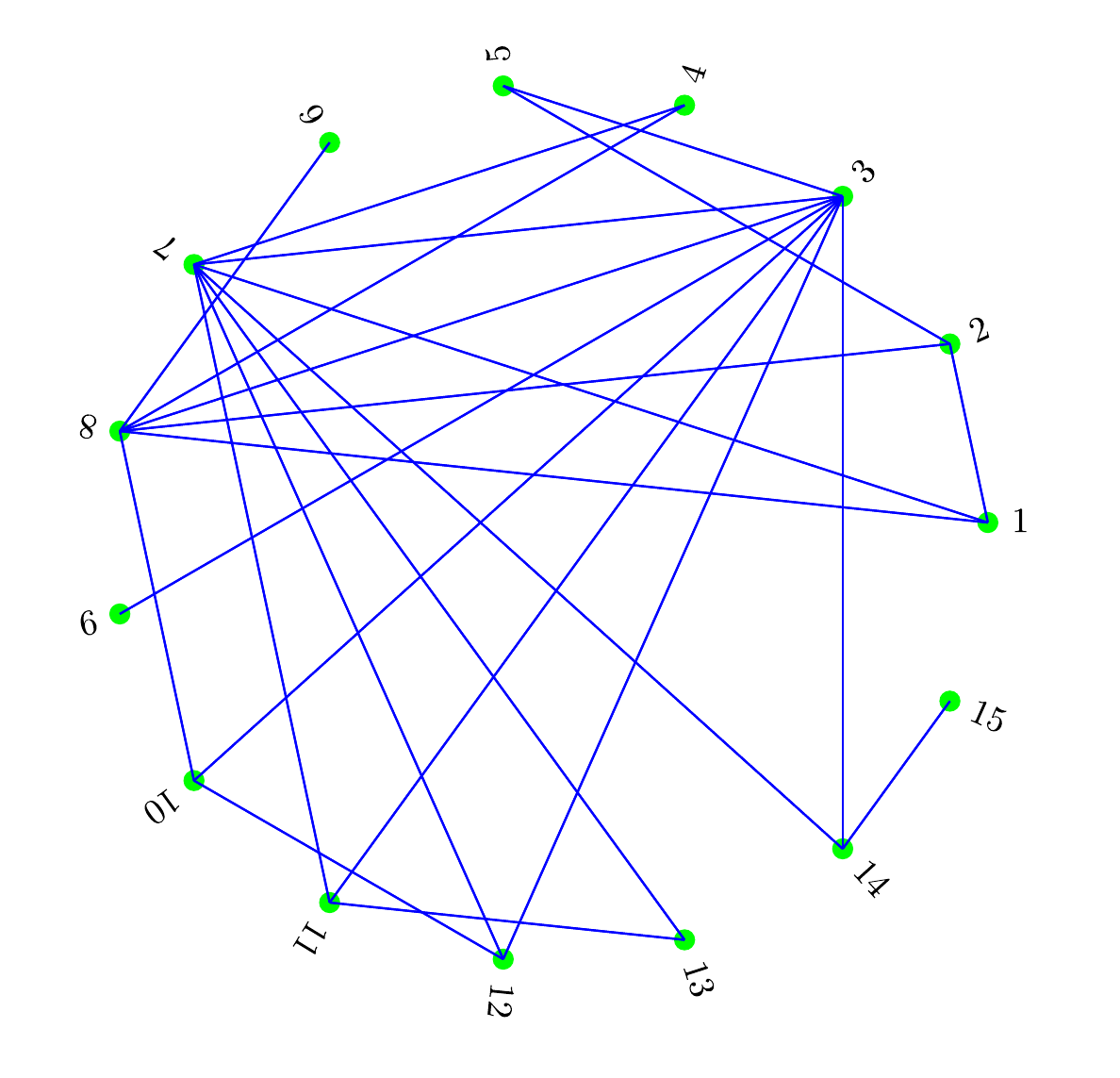}
                \caption{sparse graph}
                \label{Sparse10Graph}
        \end{subfigure}%
        \caption{Two different graph topologies of a network of 15 nodes: (a) dense and (b) sparse structures.} \label{10Graph}
\end{figure}


\section{CONCLUSIONS}
\label{sec:conclusion}

The work in this paper has been concerned with obtaining the optimal node selection for maximizing an index of observability for a network. The empirical observability Gramian has been used as a tool for improving the local observability of nonlinear systems. The observability index was chosen to be the determinant of the inverse of the observability Gramian. We proposed an optimization problem for obtaining the optimal node selection that provides the full state observability of a network. The optimization problem framed as a mixed integer nonlinear programming problem. The outer approximation algorithm was used as a relaxation method for solving a convex optimization problem with integer non-convex constrains. The outer approximation algorithm solves a relaxed approximation iteratively, and converges to the optimal value of the original mixed integer nonlinear problem. We applied the results on a model of the virus spreading process of a disease in an arbitrary network. The results show that it is possible to reconstruct the state of all the nodes by observing a select number of nodes. The result of this work is applicable on similar problems, such as the spreading of an idea or rumour through a social network like Twitter and the problem of a spreading computer virus through the World Wide Web.

In this paper, we assumed the weights are known and we do not have control on changing them. An interesting direction for future research is estimation of the infection rate and recovery rate of the nodes. Initial results in this paper indicated that if the structure of interactions between nodes and the weights of interaction between nodes (here, $\beta_i$ and $\delta_i$) are known, then by directly observing a few number of nodes in a complex network, we are able to estimate the states of the entire nodes. However, the problem of states estimation with unknown interaction weights is still a challenging problem. This analysis is an important problem when we are dealing with complex real world situations with unknown interactions (for example trusting interaction between people in a group) or when we have control on changing the weights. We are currently investigating adaptive optimal graph topologies that, given constraints on the available sensing resources, provide optimal performance regarding the observability of the system.







\bibliography{citations}
\bibliographystyle{IEEEtran}

\end{document}